\documentclass{article}
\usepackage{graphicx} 
\usepackage{amssymb,amsmath,amsthm,bbm}
\usepackage{verbatim,float,url,dsfont}
\usepackage{graphicx,subcaption,psfrag}
\usepackage{algorithm,algorithmic}
\usepackage{mathtools,enumitem}
\usepackage{multirow}
\usepackage{ragged2e}
\usepackage{xr-hyper}
\usepackage{array}

\usepackage[colorlinks=true,citecolor=blue,urlcolor=blue,linkcolor=blue]{hyperref}
\usepackage[margin=1in]{geometry}
\usepackage[round]{natbib}

\usepackage[utf8]{inputenc} 
\usepackage[T1]{fontenc}    
\usepackage{booktabs}       
\usepackage{nicefrac}         
\usepackage{microtype}      

\ifdefined\TimesFont 
\usepackage{times} 
\fi

\ifdefined\ParSkip 
\usepackage{parskip} 
\fi

\newtheorem{theorem}{Theorem}
\newtheorem{lemma}{Lemma}
\newtheorem{corollary}{Corollary}

\theoremstyle{definition}

\newtheorem*{assumption*}{\assumptionnumber}
\providecommand{\assumptionnumber}{}


\makeatletter
\newcommand*\rel@kern[1]{\kern#1\dimexpr\macc@kerna}
\newcommand*\widebar[1]{%
  \begingroup
  \def\mathaccent##1##2{%
    \rel@kern{0.8}%
    \overline{\rel@kern{-0.8}\macc@nucleus\rel@kern{0.2}}%
    \rel@kern{-0.2}%
  }%
  \macc@depth\@ne
  \let\math@bgroup\@empty \let\math@egroup\macc@set@skewchar
  \mathsurround\z@ \frozen@everymath{\mathgroup\macc@group\relax}%
  \macc@set@skewchar\relax
  \let\mathaccentV\macc@nested@a
  \macc@nested@a\relax111{#1}%
  \endgroup
}
\makeatother

\DeclareMathOperator{\Var}{Var}

\def\P{\mathbb{P}}

\def\cN{\mathcal{N}}

\title{Gaussian Rank Verification}
\date{}
\author{Jeremy Goldwasser\thanks{Department of Statistics, University of California, Berkeley; \href{jeremy\_goldwasser@berkeley.edu}{jeremy\_goldwasser@berkeley.edu}}
  \and  
  Will Fithian\thanks{Department of Statistics, University of
    California, Berkeley}
  \and 
  Giles Hooker\thanks{Department of Statistics and Data Science, University of Pennsylvania}}

\DeclareMathOperator{\barsig}{\bar{\sigma}_{1j}}
\DeclareMathOperator{\barmu}{\bar{\mu}_{1j}}
\DeclareMathOperator{\bareta}{\bar{\eta}_{1j}}
\begin{document}

\maketitle
\begin{abstract}
Statistical experiments often seek to identify random variables with the largest population means. 
This inferential task, known as rank verification, has been well-studied on Gaussian data with equal variances. 
This work provides the first treatment of the unequal variances case, utilizing ideas from the selective inference literature.
We design a hypothesis test that verifies the rank of the largest observed value without losing power due to multiple testing corrections.
This test is subsequently extended for two procedures:
Identifying some number of correctly-ordered Gaussian means, and validating the top-$K$ set.
The testing procedures are validated on NHANES survey data. 
\end{abstract}

\section{Introduction}

The goal of many statistical analyses is to understand which random variables have the highest mean. 
For example, recommender systems may aim to find the products a user has the best chance of engaging with; 
global explanations of machine learning models seek to accurately highlight the features that tend to drive its behavior; 
and biological signals like genes may be compared to identify the salient markers of patient prognosis. 
Identifying the largest means has also been studied extensively in the bandits literature, e.g. \citet{bubeck2013multiple}.
While the true means are unknown, the analyses reveal some random observation based on the data. 

Often, the data being compared are normally distributed \citep{bellcurve}. 
In linear models, the errors are typically assumed to be normal, in which case the coefficients are as well. 
In other settings like effect size estimation, normality holds asymptotically via the central limit theorem. While the variance may be the same across groups, this is not necessarily true.

Formally, suppose we observe a single independent draw from each of $d$ random variables $X_j\sim \cN(\mu_j, \sigma_j^2)$, where $\sigma_j^2$ is known for all $j \in \{1, \ldots d\}$. Borrowing from \citet{HF}, the largest observed value will hereafter be referred to as the ``winner,'' and the second-largest as the ``runner-up.'' Further, we denote the unknown object with the largest mean $j^* = \mbox{argmax}_j \mu_j$ as the ``best''.

Assume without loss of generality that the data is ordered via 
$$X_1 > X_2 > \ldots > X_d.$$

\noindent Having observed this data, we would like to draw inferences about the ranking of the means $\mu_j$. Broadly, this is the task known as rank verification. 
Throughout this work, we will assume observations are distinct. However, \citet{HF} presents a simple argument justifying the same procedures in the presence of ties. 

When all variances are equal, a level-$\alpha/2$ test between the winner and the runner-up suffices to verify the winner $X_1$ as the best \citep{gutmann1987selected, HF} --- that is, $\mu_1 > \mu_{j} $, $\forall j > 1$ --- with probability exceeding $1-\alpha$ if this test rejects. \citet{gutmann1987selected} was the first to show this result. Later works extended it for larger classes of distributions. \citet{stefansson1988confidence} proved its validity under balanced samples from log-concave location families. More recently, \citet{HF} applied it to a broad subset of the exponential family.

The validity of this test does not hold in the unequal variances case. Figure \ref{fig:unequal_vars} gives a simple example of this, visualizing the PDFs of 5 Gaussians. Here, most distributions have similar means and very small variances, with little overlapping mass. The sole exception, which has the highest mean by a small margin, is shown in red. 

Suppose the high-variance random variable is drawn below the second and third place.
This entirely plausible event is visualized with the large dots in Figure \ref{fig:unequal_vars}. 
The tails of the low-variance distributions barely overlap, 
so the p-value testing the winner and runner-up is an infinitesimal $2\times10^{-9}$. 
This rejects for all reasonable $\alpha$, falsely verifying the winner.
Indeed, the type I error rate of this procedure at $\alpha=0.05$ is at least 34.4\% in this setting (Appendix A). 

\begin{figure}
    \centering
    \includegraphics[width=\linewidth]{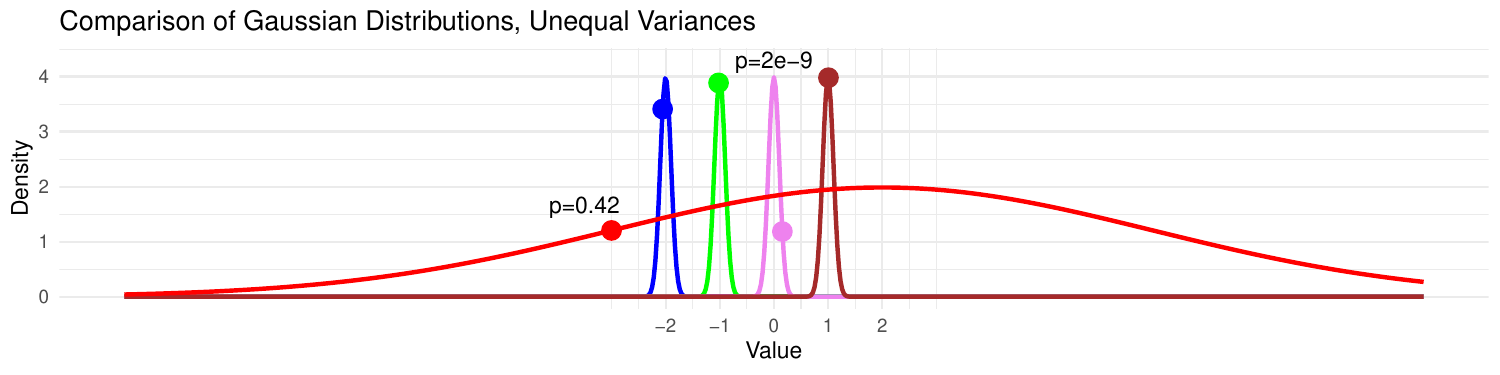}
    \caption{Testing the winner and runner-up fails in the unequal variances case. Wide distribution rescaled for ease of interpretation.}
    \label{fig:unequal_vars}
\end{figure}

In light of this, we develop a valid level-$\alpha$ hypothesis test to establish $\mu_1$ as the best. While this test considers ranks below the runner-up, it does not rely on costly multiple testing corrections that reduce the power. In Figure \ref{fig:unequal_vars}, this test's p-value is 0.42, comparing the winner against the high-variance observation. The non-significant p-value avoids a false rejection, unlike the test against the runner-up.

After establishing this test in Section \ref{sec:test}, we employ it for two related procedures, neither of which use multiple testing corrections. The first procedure identifies an integer $K\geq 0$ such that the order of the $K$ largest means matches the empirical ranking with high probability. The second is a test of the stability of the top-$K$ set, for some predetermined $K$. We introduce these testing procedures in Section \ref{sec:procedures}. Subsequently, we apply them to the NHANES health survey, and demonstrate their validity on simulated data.\footnote{\texttt{R} code to reproduce all figures and results is at \url{https://github.com/jeremy-goldwasser/Gaussian-Rankings}.}

\section{Verifying the Winner}\label{sec:test}

This section introduces a valid level-$\alpha$ test to verify the winner as the best. We state the main result, discuss its implications, then prove it using techniques from selective inference.

\subsection{Theorem Statement}
\begin{theorem}\label{thm:winner}
    Let $X_j \sim \cN(\mu_j, \sigma_j^2)$ independently for $j \in [d]=\{1,\ldots,d\}$, with realized values $x_j$. Assume $\sigma_j^2$ is known, and without loss of generality sort $X_j$ by their order statistics as above. 
    Define 
    
    $$\barmu = \frac{\sigma_j^2 x_1 + \sigma_1^2 x_j}{\sigma_1^2 + \sigma_j^2} \text{, }\barsig^2 = \frac{\sigma_1^4}{\sigma_1^2+\sigma_j^2}\text{, and } \bareta=\max(\barmu, \max_{k\neq 1,j} x_k).$$ 

    Consider the null hypothesis \mbox{$H_0:\{ \mu_1 \leq \max_{j>1} \mu_j \ \vert \ X_1\ \text{wins}\}$}, i.e. that the winner is \textit{not} the best. A valid p-value to test $H_0$ is  

    $$p_1^* = \max_{j>1} p_{1j} = \max_{j>1} 
    \frac{1-\Phi(\frac{x_1 - \barmu}{\barsig})}{1-\Phi(\frac{\bareta - \barmu}{\barsig})},$$

\noindent where $\Phi(\cdot)$ is the CDF of the standard normal distribution. The probability of falsely verifying the winner with this p-value is at most $\alpha$. 
    
\end{theorem}

The numerator of $p_{1j}$ in Theorem \ref{thm:winner} expresses the probability mass of $Z\sim\cN(0,1)$ above $\frac{x_1-\barmu}{\barsig}$. This is the test statistic of the one-sided $Z$-test with unequal variances:
 
\begin{equation}\label{eq:Z-test}
    \frac{x_1-\barmu}{\barsig} = \frac{x_1 \left(\frac{\sigma_1^2+\sigma_j^2}{\sigma_1^2+\sigma_j^2}\right) - \frac{\sigma_j^2 x_1 + \sigma_1^2 x_j}{\sigma_1^2 + \sigma_j^2}}{\frac{\sigma_1^2}{\sqrt{\sigma_1^2+\sigma_j^2}}} = \frac{\sigma_1^2 x_1 - \sigma_1^2 x_j}{\sigma_1^2\sqrt{\sigma_1^2+\sigma_j^2}} = \frac{ x_1 - x_j}{\sqrt{\sigma_1^2+\sigma_j^2}}.
\end{equation}

 Next, consider the setting in which $\barmu > \max_{k\neq 1,j} x_k$. 
 Then the denominator of $p_{1j}$ integrates the standard normal from its mean, $0$. 
 By symmetry, this is also the median, so the integral equals $\frac{1}{2}$. 
 As a result, $p_{1j}$ is significant whenever the $Z$-statistic in the numerator exceeds the $1-\alpha/2$ quantile. 
This reduction occurs for $j=2$, testing the winner against the runner-up:

$$\bar\mu_{12} = \frac{\sigma_2^2 x_1 + \sigma_1^2 x_2}{\sigma_1^2+\sigma_2^2} > \frac{\sigma_2^2 x_2 + \sigma_1^2 x_2}{\sigma_1^2+\sigma_2^2} = x_2 > \max_{k>2} x_k.$$

In practice, $j = 2$ often has the largest p-value amongst all $p_{1j}$. 
Intuitively, Equation~\eqref{eq:Z-test} highlights the inverse relation between the distance $x_1-x_j$ and the numerator of $p_{1j}$; $x_2$ is closest to $x_1$, so its p-values are likely to be larger. 
Indeed, \citet{HF} showed $p_1^*=p_{12}$ in the special case where all variances are equal. 
Our test in Theorem \ref{thm:winner} can thus be understood as generalizing the equal-variance level $\alpha/2$ $Z$-test between the winner and runner-up.

Nevertheless, the runner-up does not necessarily have the largest p-value.
In fact, $p_{1j}$ is monotonically increasing in the variance $\sigma_j^2$. 
This protects against situations like the one described in Figure \ref{fig:unequal_vars}, in which a lower observed rank has larger variance than the runner-up. 

\subsection{Proof of Theorem \ref{thm:winner}}

Verifying the winner entails demonstrating it is sufficiently unlikely to have won by its margin if its mean were \textit{not} the highest. The procedure we construct to demonstrate needs to control the type I error rate --- the probability of falsely verifying the wrong winner. 
Letting $j^*$ be the index of the highest mean, this is

$$\P(\text{Type I error}) =\P(\text{falsely verify winner}) = \P(\bigcup_{j\neq j^*} \text{$j$ wins and is verified} ).$$
We can upper bound its error rate with a union bound:

\begin{align*}
    \P(\bigcup_{j\neq j^*} \text{$j$ wins and is verified } ) \leq 
\sum_{j\neq j^*} \P(\text{$j$ wins and is verified } ) \\
= \sum_{j\neq j^*} \P(\text{verify $j$ }\vert \text{$j$ wins})\P(\text{$j$ wins}).
\end{align*}

$\P(j\text{ wins})$ is nonnegative and sums to $1$. Consequently, a procedure that ensures \mbox{$\P(\text{verify $j$ }\vert \text{$j$ wins})\leq\alpha$} for all $j\neq j^*$ controls the type I error rate at level $\alpha$. We will achieve this with a valid level-$\alpha$ test of \mbox{$H_0:\{\mu_1 \leq \max_{j>1} \mu_j \ \vert \ X_1\ \text{wins}\}$}.

To test this null, we employ the following lemma, from \citet{Berger1982}.
\begin{lemma}\label{lemma:union}
    Let $p_{1j}$ be valid p-values for null hypotheses $H_{01j}$. Then $p^* = \max_j p_{1j}$ is a valid p-value for the union null $H_0 = \bigcup H_{01j}$.
\end{lemma}
\begin{proof}[Proof of Lemma \ref{lemma:union}]
When $H_0$ is true, at least one of the nulls $H_{01j}$ is true. Without loss of generality, assume $H_{01}$ is true. By construction, $p^* \geq p_1$, so $\P(p^* \leq \alpha) \leq \P(p_1 \leq \alpha)$. $p_1$ is defined to be a valid p-value, so $\P(p_1 \leq \alpha) \leq \alpha$. This implies the validity of $p^*$.
\end{proof}

In our setting, $H_0$ can be expressed as a union null, via 
\begin{equation}
    H_0: \bigcup_{j>1} \underbrace{\left\{\mu_1 \leq \mu_j \ \vert \ X_1 > \max_{k>1} X_k \right\}}_{H_{01j}}.
\end{equation}
Therefore by Lemma \ref{lemma:union}, a valid p-value for $H_0$ is the maximum p-value $p_{1j}$ testing $H_{01j}$ over all $j>1$. The remainder of this proof concerns constructing these valid p-values, employing a similar strategy as \citet{HF}. 

If $X_1$ were not the winner, these tests would not be conducted in the first place.
However, classical statistical methods fail when the tests to conduct are selected \textit{after} examining the data.
Instead, the test must condition on $X_1$ being the largest, per the selective inference framework of \citet{fithian_sun_taylor}.

Let $A_1$ denote this selection event, and let $\phi$ be the multivariate normal density. The joint distribution of $X$ is

$$X\vert A_1 \propto \phi(X)\mathds{1}_{A_1},$$
 
\noindent where $\propto$ indicates the density is proportional in $X$ up to a constant normalizing factor. We proceed by conditioning on more variables in addition to $A_1$. 
The ensuing test of $H_{01j}$ will be valid conditional on both $A_1$ and any possible realization of these extra variables.
As a result, it is still valid after marginalizing them out so only $A_1$ remains --- the test we need to be valid.

In particular, we first condition on all observed values besides $1$ and $j$. By independence, 
\begin{align}\label{eq:exp_form}
    X_{1,j}\vert A_1, X_{k\neq 1, j} &\propto \phi(X_1, X_2)\mathds{1}_{A_1}\nonumber\\
    &\propto \exp\left[{-\frac{1}{2\sigma_1^2}(X_1-\mu_1)^2-\frac{1}{2\sigma_2^2}(X_2-\mu_2)^2}\right]\mathds{1}_{A_1}\nonumber\\
    &\propto \exp\left[\frac{X_1\mu_1}{\sigma_1^2}+\frac{X_2\mu_2}{\sigma_1^2}-\frac{X_1^2}{2\sigma_1^2}-\frac{X_2^2}{2\sigma_2^2}\right]\mathds{1}_{A_1}\nonumber\\
    &= \exp
        \left[
        \underbrace{\frac{X_1}{\sigma_1^2}}_{T(X)}
        \underbrace{(\mu_1-\mu_j)}_{\theta}
        +\underbrace{\left(\frac{X_1}{\sigma_1^2}
        +\frac{X_j}{\sigma_j^2}\right)}_{U(X)}
        \underbrace{\mu_j}_{\lambda} 
        -\frac{X_1^2}{2\sigma_1^2}
        -\frac{X_j^2}{2\sigma_j^2}
        \right]
        \textbf{1}_{A_1}.
\end{align}

Under the null hypothesis, $\mu_1-\mu_j\leq 0$.
In Equation~\eqref{eq:exp_form}, $\mu_1-\mu_j$ appears in $\theta$, divided by known constant $\sigma_1^2$. 
However, we cannot immediately conduct inference on $\theta$ via its sufficient statistic $T(X)$. This is due to the term involving the nuisance parameter $\lambda=\mu_j$.
To evade its influence, we condition on the statistic \mbox{$U(X)=\frac{X_1}{\sigma_1^2}+\frac{X_j}{\sigma_j^2}$} taking the realized value, \mbox{$u=\frac{x_1}{\sigma_1^2}+\frac{x_j}{\sigma_j^2}$}. 
Upon doing so, the nuisance parameter factors out of the exponential, thereby isolating $\theta$:

\begin{equation}\label{eq:nuisance_condition}
    X_{1}\vert \{A_1, X_{k\neq 1, j}, U\} \propto \exp\left[\frac{X_1}{\sigma_1^2}(\mu_1-\mu_j) -\frac{X_1^2}{2\sigma_1^2} -\frac{X_j^2}{2\sigma_j^2}\right]\textbf{1}_{A_1}.
\end{equation}

The most powerful level-$\alpha$ test is likeliest to reject when the null is false. This is attained at the boundary case $\mu_1=\mu_j$, equivalently where $\theta=0$. Performing the null reduction simplifies Equation~\eqref{eq:nuisance_condition} to

\begin{equation}\label{eq:null_reduction}
    X_1, X_j\vert \{A_1, X_{k\neq 1, j}, U\} \propto \exp\left[-\frac{X_1^2}{2\sigma_1^2} -\frac{X_j^2}{2\sigma_j^2}\right]\textbf{1}_{A_1}.
\end{equation}

Note the condition $U(X)=u$ can be rearranged as \mbox{$X_j = \sigma_j^2(u-\frac{X_1}{\sigma_1^2})$}. As a result, the density in Equation~\eqref{eq:null_reduction} may be posed solely in terms of $X_1$.
In fact, completing the square in the exponential reveals that the conditional distribution of $X_1$ is a truncated normal.

\begin{align}\label{eq:complete_square}
    X_{1}\vert \{A_1, X_{k\neq 1, j}, U\} &\propto \exp\left[{-\frac{X_1^2}{2\sigma_1^2} -\frac{\left(\sigma_j^2(u-\frac{X_1}{\sigma_1^2}\right)^2}{2\sigma_j^2}}\right]\textbf{1}_{A_1}\nonumber\\
    &\propto \exp\left[{ -\frac{X_1^2}{2\sigma_1^2} -\frac{\sigma_j^2 X_1^2}{2\sigma_1^4} + \frac{\sigma_j^2 u X_1}{\sigma_1^2}}\right]\textbf{1}_{A_1}\nonumber\\
    &\propto \exp\left[ \left(-\frac{\sigma_1^2+\sigma_j^2}{2\sigma_1^4}\right)X_1^2 + \left(\frac{\sigma_j^2 u }{\sigma_1^2}\right)X_1\right]\textbf{1}_{A_1}\nonumber\\
    &\propto \exp\left[-\frac{\sigma_1^2+\sigma_j^2}{2\sigma_1^4}\left(X_1^2 - \frac{2\sigma_1^2\sigma_j^2 u }{\sigma_1^2+\sigma_j^2} X_1\right) \right]\textbf{1}_{A_1}\nonumber\\
    &\propto \exp\left[-\frac{\sigma_1^2+\sigma_j^2}{2\sigma_1^4}\left(X_1 - \frac{\sigma_1^2\sigma_j^2 u }{\sigma_1^2+\sigma_j^2}\right)^2\right]\textbf{1}_{A_1}.
\end{align}


Analyzing Equation~\eqref{eq:complete_square}, the truncated normal has mean and variance parameters $\barmu = \frac{\sigma_j^2 x_1 + \sigma_1^2 x_j}{\sigma_1^2 + \sigma_j^2}$ and $\barsig^2 = \frac{\sigma_1^4}{\sigma_1^2+\sigma_j^2}$. 
The truncation event $A_1$ may be expressed using the conditioned variables $X_{k}=x_{k} \;\forall k\neq 1, j$ and $U(X)=u$. 



\begin{align}\label{eq:selection}
    A_1 &= \{X_1 \text{ wins}\}\nonumber \\
        &= \{X_1 > X_j \;\bigcap\; X_1 > \max_{k\neq 1, j} X_k \}\nonumber \\
    &= \{X_1 > \sigma_j^2\left(u-\frac{X_1}{\sigma_1^2}\right)\;\bigcap\; X_1 > \max_{k\neq 1, j} x_k \}\nonumber \\
    &= \{ X_1 > \frac{\sigma_1^2 \sigma_j^2 u}{\sigma_1^2 + \sigma_j^2}\;\bigcap\; X_1 > \max_{k\neq 1, j} x_k\}\nonumber \\
    &= \{X_1 > \underbrace{\max( \barmu, \max_{k\neq 1, j} x_k )}_{\bareta}\}
\end{align}

The p-value $p_{1j}$ is the tail mass of this truncated normal distribution. The tail is the values of $X_1$ greater than the observed value, $x_1$. 
Figure \ref{fig:segments} displays how the p-value relates to the truncation event. 
To compute $p_{1j}$, take the tail mass of the normal distribution, divided by the total mass after truncation event. 
Expressing in terms of $\barmu, \barsig$, and $\bareta$, this is

\begin{equation}\label{eq:pval}
    p_{1j} = \frac{1-\Phi(\frac{x_1 - \barmu}{\barsig})}{1-\Phi(\frac{\bareta - \barmu}{\barsig})}.
\end{equation}

Finally, use $p_1^* = \max_j p_{1j}$ as the p-value. When it rejects at level $\alpha$, the winner is verified as the best. 

$\hfill\square$

\begin{figure}[H]
     \centering
     \begin{subfigure}[b]{0.49\linewidth} 
         \centering
         \includegraphics[width=\linewidth]{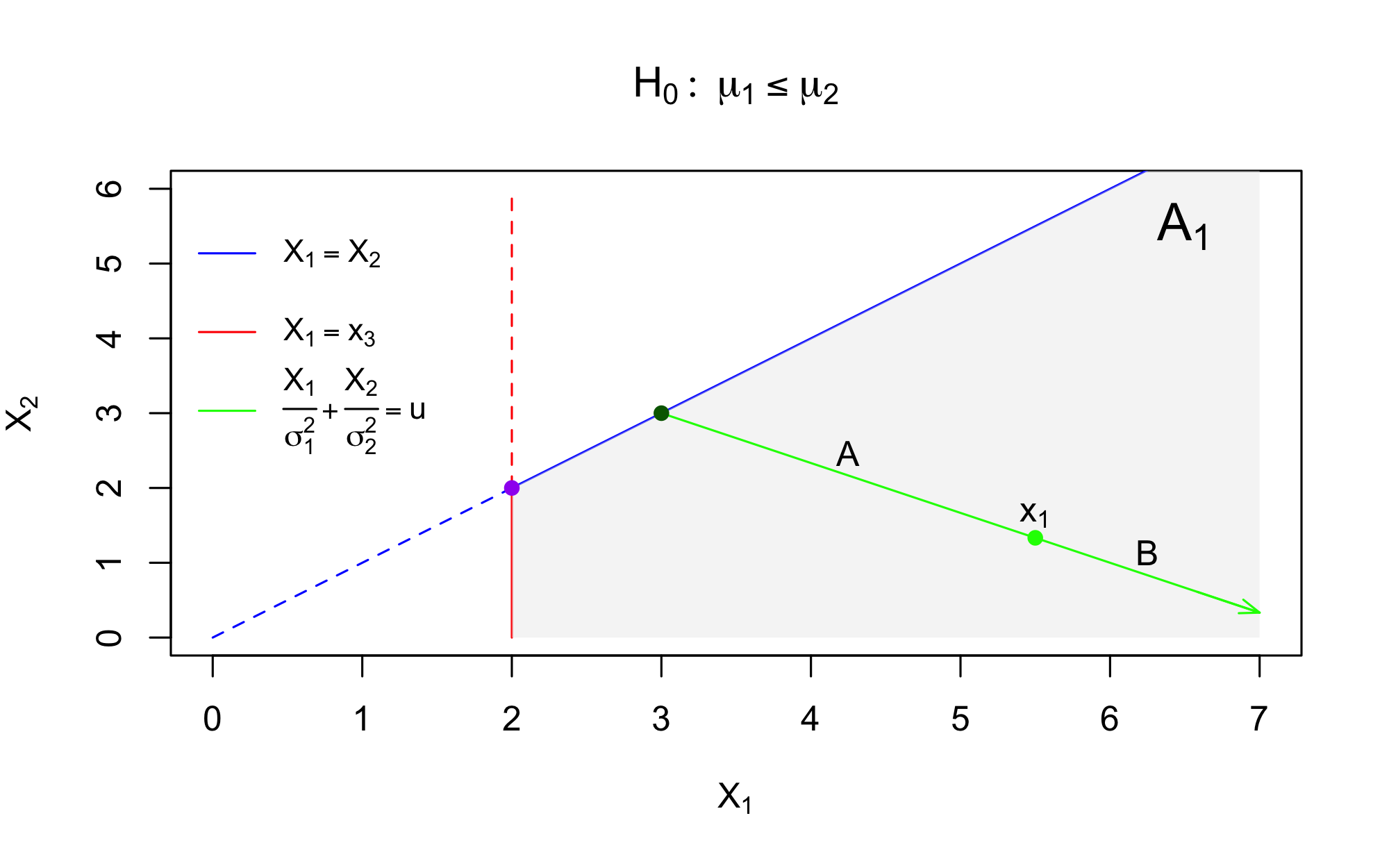}
         \caption{Testing winner against runner-up.}
     \end{subfigure}
     \hfill 
     \begin{subfigure}[b]{0.49\linewidth}
         \centering
         \includegraphics[width=\linewidth]{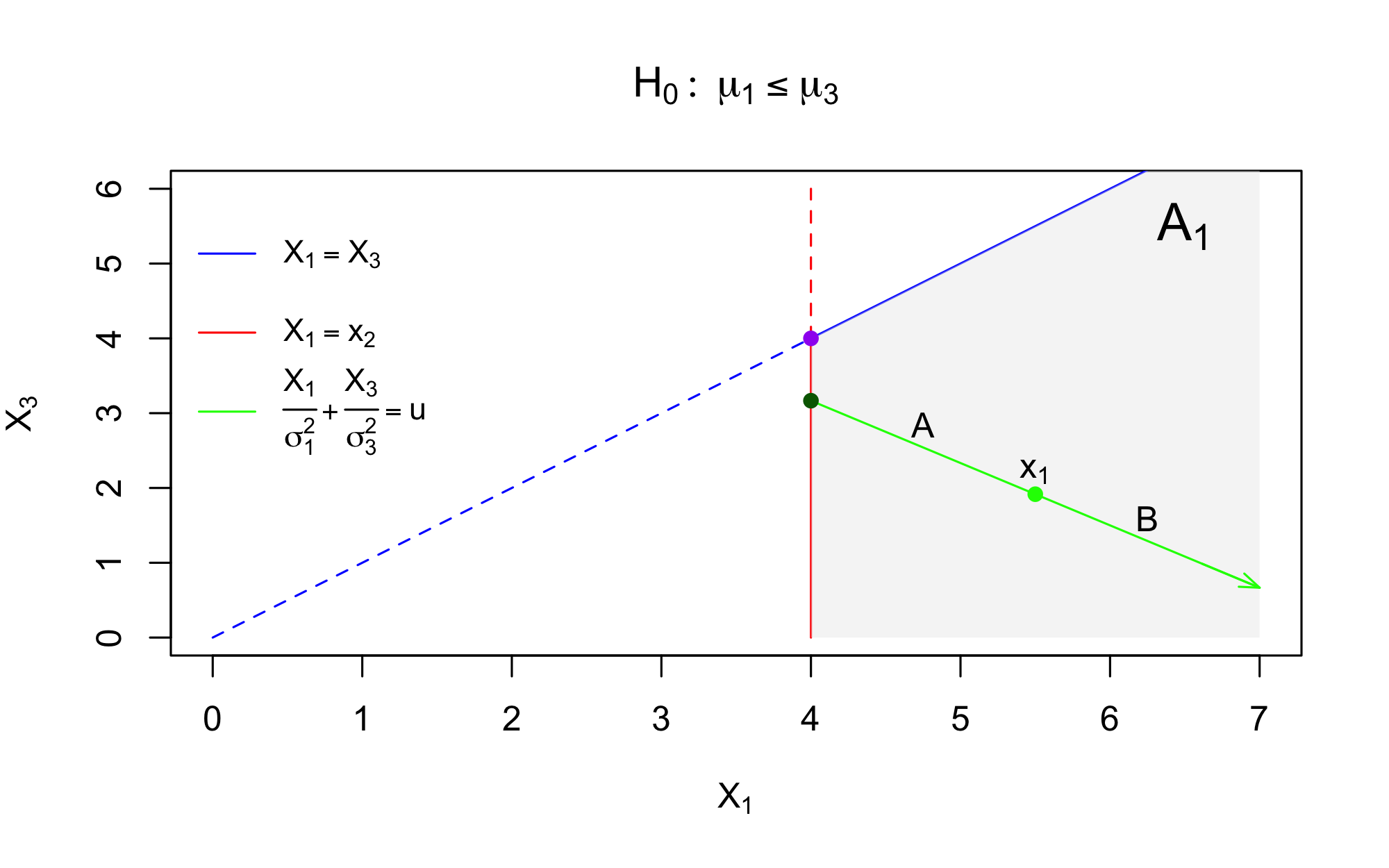}
         \caption{Testing winner against third place.}
     \end{subfigure}
    \caption{Tests to verify winner. Shaded region depicts truncation event; p-value is mass of segments $B/(A+B)$.}
    \label{fig:segments}
\end{figure}

\paragraph{Commentary and Modification.}

It is no coincidence that the post-selection distribution is a truncated normal.
This is the distribution of any linear contrast of Gaussian data conditioned on a polyhedral selection event \citep{Tibshirani02042016, Lee_2016, PNAS}. 
In our setting, selection is polyhedral: $X_1$ winning can be expressed with the linear event that $X_1 - X_j > 0$ for all $j>1$. Indeed, the same test can be equivalently derived via the Polyhedral Lemma of \citet{Tibshirani02042016}. 

This finding has been applied to perform post-selection inference in a variety of contexts.
These include analyzing coefficients of stepwise regression \citep{Tibshirani02042016} and the Lasso \citep{Lee_2016}; rank estimation via principal components analysis \citep{choiPCA}; changepoint detection \citep{Hyun2021}; and many others.
Perhaps the closest application to ours is \citet{reid2015}, which utilizes it for top-$K$ mean estimation.

There may be settings in which it is of interest to validate the lowest rank, as opposed to the highest.  This changes the test in Theorem \ref{thm:winner}, since we now are interested in the lower tail of the tested variable. 
Let $1$ here indicate the lowest order statistic, and define $\gamma_j=\min(\barmu, \min_{k\neq 1,j} x_k)$. 
The equation for the p-value is is changed from Equation~\eqref{eq:pval} to

\begin{equation}\label{eq:pval_low}
    \tilde{p}_{1j} = \frac{\Phi(\frac{x_1 - \barmu}{\barsig})}{\Phi(\frac{\gamma_j - \barmu}{\barsig})}.
\end{equation}

\section{Top-$K$ Methods}\label{sec:procedures}

We apply Theorem \ref{thm:winner} in the service of verifying the top $K$ ranks.

\subsection{Ordered Ranking}\label{sec:ranks}

Of interest may be the ranking of the top $K$ means, rather than just the first. Having observed the ordering $X_1 > X_2 > \ldots > X_d$, the aim is to identify some number of ranks $K\geq 0$ such that with probability exceeding $1-\alpha$,

$$\mu_1 > \mu_2 > \ldots > \mu_K > \max_{j>K} \mu_j.$$

\noindent\textbf{Procedure 1}. Starting with the first order statistic, conduct the level-$\alpha$ verification test from Theorem \ref{thm:winner}. When the test rejects, reiterate with the subsequent order statistic, testing it against all the ranks it exceeds. As soon as a test fails to reject, return the number of rejected tests $K\geq 0$.

\begin{corollary}\label{cor:topKranks}
    Procedure 1 is a valid level-$\alpha$ procedure. The family-wise error rate (FWER) --- the probability that one or more of its top-$K$ ranks is out of order --- is at most $\alpha$. 
\end{corollary}

\begin{proof}
    Corollary \ref{cor:topKranks} holds by an identical argument as Section 5 of \citet{HF}. Rather than recreate this in its entirety, we briefly summarize its main argument. For greater detail, we refer the reader to the original manuscript.

    The $k^\text{th}$ test in Procedure 1 takes the maximum p-value $p_{kj}$, comparing $X_k$ to every lower-ranked $X_j$. 
    The proof identifies more liberal p-values than $p_{kj}$ which nevertheless control FWER.
    In addition to the selection event~\eqref{eq:selection}, these p-values also condition on the event that $X_k$ is less than $X_{k-1}$. 
    They are designed to test a wider null hypothesis: That some element(s) of the top-$k$ are ordered improperly.

    These hypotheses are nested, such that all nulls following the first true null must also be true. Given fixed, nested null hypotheses, stopping at the first failure to reject is a valid level-$\alpha$ procedure. This is a classical result dating back to \citet{marcus1976}. More recently, \citet{will_ryan} showed it can be adapted for random nested nulls by adding conditions to fix each null. In the context of rank verification, this is achieved by conditioning on $X_k$ being less than $X_{k-1}$. 
\end{proof}

As noted previously, one may hope to validate the $K$ lowest ranks. 
To do so, modify Procedure 2 to validate whether a given rank is \textit{lower} than all preceding ranks, iterating from the lowest rank. Each test rejects if the p-value from Equation~\eqref{eq:pval_low} is below $\alpha$.

\subsection{Top-$K$ Set}

It may be of practical interest to identify the $K$ highest-ranking elements, without necessarily establishing their ranking within that set. In this section, we introduce a simple testing procedure to do so, which is generally less stringent than the test in Section \ref{sec:ranks}. 

\bigskip

\noindent\textbf{Procedure 2}. 
For each $k$ in 1 to $K$, test element $k$ versus all $d-K$ elements after $K$ with Theorem \ref{thm:winner}. If all $K$ tests reject, then declare the top-$K$ observed set is correct with probability at least $1-\alpha$.

\begin{corollary}\label{cor:topKset}
    Procedure 2 is a valid level-$\alpha$ procedure. The probability that all tests reject yet the top-$K$ observed set is incorrect is at most $\alpha$.
\end{corollary}

\begin{proof}
    Procedure 2 tests the null hypothesis that some element of the top-$K$ set does not belong. This can be written as a union over $K$ null hypotheses:
    
\begin{align*}
    H_0:&\;\exists\  j \in \{1,\ldots,K\}\ \text{ s.t. }\ \mu_j \leq \max_{k > K} \mu_j \ \vert \ X_j\ \text{wins}\ \\
    \Longleftrightarrow \quad&\bigcup_{j=1}^K \underbrace{\left\{\mu_j \leq  \max_{k > K} \mu_k \ \vert \ X_j\ > \max_{\ell > K}X_\ell\ \right\}}_{H_{0j}}.
\end{align*}
    
By Lemma \ref{lemma:union}, a valid p-value to test this null is the maximal $p_j$ testing $H_{0j}$. 
Each null hypothesis $H_{0j}$ is itself a union null of the form

\begin{equation*}
        H_{0j}: \bigcup_{k=K+1}^d \underbrace{\left\{\mu_j \leq  \mu_k \ \vert \ X_j\ > \max_{\ell>K}X_\ell\ \right\}}_{H_{0jk}}.
\end{equation*}

By Theorem 1, a valid choice of $p_j$ is the maximal p-value $p_{jk}$ testing null $H_{0jk}$. 
Therefore, a valid p-value for $H_0$ is 
$$p^* = \max_{j \in [1\ :\ K]}\ \max_{k \in [K+1\ :\ d]} p_{jk}.$$

\end{proof}


\section{Empirical Study}

\subsection{Validation}

\begin{figure}
    \centering
    \includegraphics[width=\linewidth]{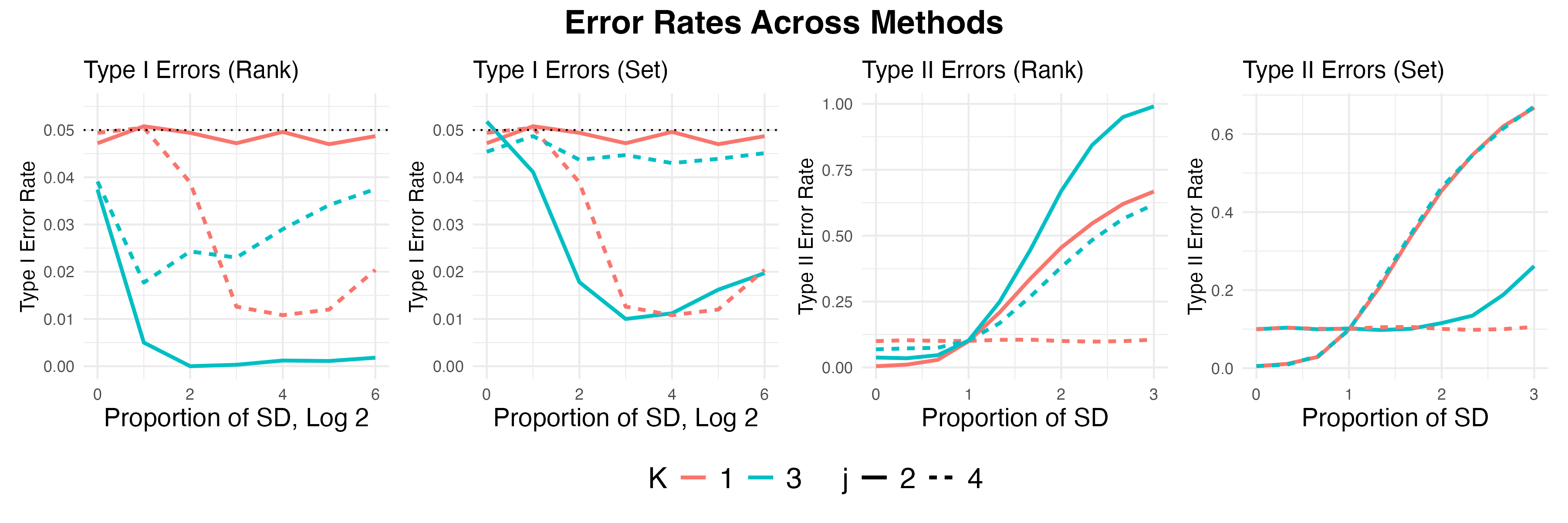}
    \caption{Type I and II errors of our top-$K$ ranking and set methods on simulated data. All variables have equal variance except one, whose standard deviation rises as a proportion of the others'.}
    \label{fig:sim}
\end{figure}

We demonstrated the validity of our top-$K$ ranking and set methods on simulated data. Our simulations contained 5 Gaussian random variables with a range of variances. 
Their means depended on whether we assessed type I or II errors.
To assess type II error -- when the observed top $K$ is in order, but is not verified -- the means were spaced evenly. 
To assess type I error, we let $\mu_K = \mu_{K+1}$; in this context there is no true top $K$, so the procedure would be mistaken to verify.

We identified the standard deviations $\bar\sigma$ for which each procedure had power 0.9 in the evenly-spaced setting. 
Among the simulated draws for which the top-$K$ variables are correct, power is calculated as the proportion in which for which the set or ranking tests reject.
Next, we simulated data for which all but one of the random variables had this default standard devation. 
The atypical random variable has population rank $j=2$ or 4.
We chose a number of values for its standard deviation $\sigma_j$. For our experiments analyzing type I error,  $\sigma_j$ grew exponentially from $\bar\sigma=2^0\bar\sigma$ to $2^6\bar\sigma$, with evenly-spaced exponents. The range on our type II experiments was more modest, and realistic. $\sigma_j$ ranged from 0, i.e. producing a constant value, to $3\bar\sigma$.

Figure \ref{fig:sim} displays the results from 10,000 draws. 
The leftward two panels show that the type I errors virtually never violate the specified threshold. This bolsters the theoretical validity of our procedures. 
However, the error rates are often quite close to $\alpha=0.05$, indicating our tests are well-calibrated and should thus have relatively high power. 
In addition, Figure \ref{fig:type_i_even} shows type I error in the evenly-spaced setting. Calculated as the proportion of rejections amongst misordered draws, these error rates are more conservative when $\sigma_j$ is small. This is because misorderings are generally closer, due to the gap between $\mu_K$ and $\mu_{K+1}$.

%

The type II error rate, equal to 1 minus the power, generally rises with $\sigma_j$. That our procedures are more cautious with large $\sigma_j$ makes intuitive sense, as it increases the likelihood that the observed top $K$ is incorrect. In the most extreme case, the ranking method rarely verifies the top 3 ranks when second place is high-variance, since both the first and third ranks could be misordered. In general, however, the power remains reasonable. A notable exception, in dotted red, is for $K=1$ when $\sigma_4$ is inflated. $\mu_4$ is far enough from $\mu_1$ that this has virtually no effect on the type II error rate.

\subsection{NHANES Application}\label{sec:nhanes}
We applied these testing methods to the National Health and Nutrition Examination Survey, detailed in Appendix B.
In particular, we studied health outcomes by education status, stratified into five groups (Table S1). 
Respondents reported their highest level of education as 8th grade, 9th-11th grade, high school graduate, some college, or college graduate.
The health outcomes we examined were log income, hours of sleep per night, and days of bad mental health per month. 
The number of samples in each group is large, ranging between roughly 450 and 2250. 
Therefore by the central limit theorem, the sample means within each group are asymptotically normal. 
Moreover, the sample variances are close to the theoretical variance, so we treat the $\sigma_j^2$ as known.

\begin{table}[h!]
\centering
\begin{tabular}{|c|c|c|c|}
\hline
\textbf{Outcome} & \textbf{Stable Top-$K$} & \textbf{Stable Bottom-$K$} & \textbf{Top-$3$ p-value} \\ \hline
Log Income & 5 & 5 &0.00 \\ \hline
Sleep Hours & 1 ($p_1^*=.055$) & 0 ($p_1^*=.954$) & 0.17 \\ \hline
Days Bad Mental Health & 0 ($p_1^*=.103$) & 1 ($p_1^*=.002$) & 0.38 \\ \hline
\end{tabular}
\caption{Results of ranking and set procedures, NHANES. Ranking procedures use $\alpha=0.1$, though p-values are also reported.}
\label{tab:nhanes-results}
\end{table}

We used our ranking procedures to validate how education relates to these health outcomes (Table \ref{tab:nhanes-results}). This requires the reasonable assumption of independence between groups. First, we validated the highest rankings with Procedure 1, at various levels of $\alpha$. For log income with any reasonable $\alpha$ (e.g. $0.01$ and above), our test rejects at all five ranks. This validates the observed, intuitive finding that average income increases with each step of educational attainment. 

Procedure 1 validates fewer ranks on the other outcomes. 
On average, college graduates have 0.15 more hours of sleep than the second-highest group.
This finding is relatively significant, rejecting at $\alpha=0.1$ but not $\alpha=0.05$.
None of subsequent top ranks reject.
Going in the opposite direction, the lowest ranks are also insignificant.
The fewest average hours of sleep are extremely close to one another (6.77 and 6.78).

People who finished college have, on average, the fewest days of bad mental health per month. 
This is by a nontrivial margin: 3.13, versus 4.20 for the second-lowest group.
The ranking procedure validates this finding ($p=0.002$), but no more bottom ranks.
People educated through 9-11th grade have the most days of bad mental health.
The next-worst group is people who finished only 8th grade, with 4.66 compared to 5.61.
This is perhaps counterintuitive, as one might expect mental health to trend uniformly with education.
Indeed, Procedure 1 does not verify this top rank at $\alpha=0.1$ (p=$0.103$), or any that follow with $\alpha=0.2$. 


We also attempt to validate sets of three educational ranks. This set denotes whether or not one has attended at least some college. 
On log income, Procedure 2 rejected at reasonable $\alpha$. This is unsurprising, since all tests reject on the more stringent ordering task. 
For sleep, the top-three p-value is $0.17$. 
The p-value is even further for rejecting on mental health, at $0.38$.

\section{Discussion}

This work analyzes the problem
of rank verification for independent heteroskedastic Gaussians, without relying on multiple testing corrections. 
The hypothesis test in Theorem \ref{thm:winner} verifies the winner as the best, 
reducing in some cases to a level-$\alpha/2$ $Z$-test. 
It is extended for two procedures that validate the ranking of multiple top observations. 

Procedure 1 controls the FWER of the top $K\geq 0$ verified ranks.
A less conservative alternative is to control the false discovery rate, or FDR. To do so, sequential stopping rules from works like \citet{GSell2016Sequential} and \citet{li_barber} may be applicable.
Similarly, Procedure 2 either verifies the entire top-$K$ set, or nothing if it fails to reject.
FDR-controlling methods such as the BH procedure may be employed to identify members of the top $K$ set with high probability. 


The justification for Gaussian data is often made through a central limit theorem. 
In this setting, the sample variance may be treated as the population quantity.
It converges at an  $O(n^{-\frac{1}{2}})$ rate, so doing so is reasonable when the sample size is sufficiently large, as in the NHANES experiments. 
Otherwise, use of our procedures may be imprecise.
Future work in the same vein as \citet{asymptotic_bootstrap} could analyze their asymptotic validity.
Alternative extensions could propose new methodology to handle the unknown variance case.

Power, particularly as $d$ grows, is another open question. 
Verifying successive ranks is less likely when means are packed more compactly.
Our work generalizes that of \citet{HF}, which demonstrated superior power to existing methods.
Lastly, future work may also explore the case of correlated Gaussians. This may be challenging, however, as conditioning on $X_{k\neq 1,j}$ does not remove their means $\mu_k$ from the density.

\newpage
\bibliographystyle{apalike}
\bibliography{refs}

\appendix

\section{Type I Error Rate, Figure 1}\label{apx:bound}

Let $A, B, \ldots, E$ denote the rankings with $\mu_A=2$ and $\mu_E=-2$. 
A type I error occurs when $X_A$ is not the highest, yet the winner is verified. 
This is likeliest to occur when $A$ is lower than $B$ and $C$, and $X_B$ is significantly above $X_C$. 
Defining $\sigma_{BC}^2 = \Var(X_B-X_C) = \sigma_B^2+\sigma_C^2$, the type I error rate has the following lower bound:

$$\P(\text{Type I Error}) \geq \P\left(X_B > X_A, X_C > X_A, \frac{X_B-X_C}{\sigma_{BC}}>Z_{1-\alpha}\right).$$

Next, reparameterize $Q=X_B-X_C$, with realization $q=x_B-x_C$. 
The rejection event is now $Q >  Z_{1-\alpha}\sigma_{BC}$. 
Denote the density of $\cN(\mu_i, \sigma_i^2)$ at $x_i$ as $f_{i}(x_i) = \frac{1}{\sigma_i}\varphi\left(\frac{x_i-\mu_i}{\sigma_i}\right)$. 
Note $f_Q$ has mean $\mu_B-\mu_C$ and variance $\sigma_{BC}$. Then,

$$\P\left(X_B > X_A, X_C > X_A, Q>Z_{1-\alpha}\sigma_{BC}\right) = \int_{-\infty}^\infty \int_{Z_{1-\alpha}\sigma_{BC}}^\infty \int_{-\infty}^{x_C} f_A(x_A)f_C(x_C) f_Q(q) dA dQ dC.$$

We evaluate this expression via numerical integration using the \texttt{pracma} package in \texttt{R}. Each integral is discretized into 100 parts. In lieu of $\pm \infty$, we take 4 standard deviations above or below the mean. Plugging in the means and variances in Figure \ref{fig:unequal_vars}, the resulting type I error rate at $\alpha=0.05$ is at least 34.4\%. 

\section{Additional data}\label{apx:nhanes}

\begin{figure}
    \centering
    \includegraphics[width=.55\linewidth]{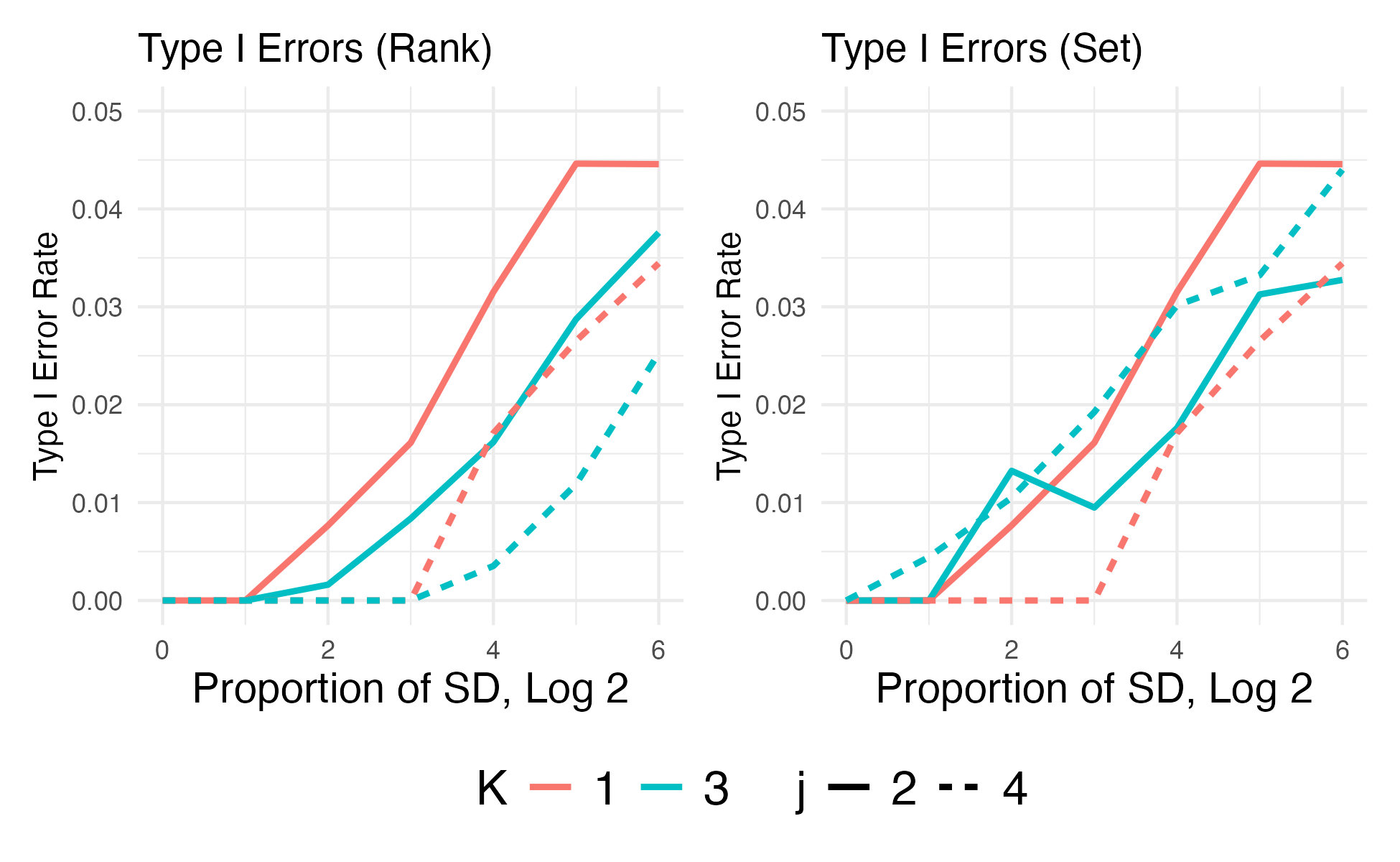}
    \caption{Type I Errors, Evenly-Spaced Means.}
    \label{fig:type_i_even}
\end{figure}

\textbf{Validation.} Figure \ref{fig:type_i_even} plots the type I error rates of our methods, given evenly-spaced means and inflating variance.
\\
\noindent \textbf{NHANES.} NHANES surveyed 10,000 representative Americans on their health and nutrition.
Significant public health findings should generalize to the U.S. civilian population from 2009-2012.

\begin{table}[h]
\centering
\begin{tabular}{|c|c|c|c|c|}
\hline
\textbf{Education}      & \textbf{N} & \textbf{Log Income} & \textbf{Hours of Sleep} & \textbf{Days of Bad Mental Health} \\ \hline
8th Grade      & 451 & 10.04     & 6.89       & 4.66                \\ \hline
9 - 11th Grade & 888 & 10.28     & 6.77       & 5.61                \\ \hline
High School    & 1517 & 10.52     & 6.78       & 4.20                \\ \hline
Some College   & 2267 & 10.74     & 6.87       & 4.64                \\ \hline
College Grad   & 2098 & 11.14     & 7.04       & 3.13                \\ \hline
\end{tabular}
\caption{Summary of NHANES data.}
\label{tab:nhanes}
\end{table}

\end{document}